\documentclass{article}
\usepackage{amssymb,amsthm,amsmath,amsfonts}
\usepackage{epsfig}

\theoremstyle{plain}
\begingroup
 
\newtheorem{theorem}{Theorem} 
\newtheorem{corollary}{Corollary}

\endgroup
\newcommand{\nwc}{\newcommand}
\nwc{\bit}{\begin{itemize}}
\nwc{\eit}{\end{itemize}}
\nwc{\Levy}{L\'evy}
\nwc{\LK}{L\'evy-Khintchine}
\nwc{\LI}{L\'evy-It\^{o}}
\nwc{\CH}{Cole-Hopf}
\nwc{\Holder}{H\"{o}lder}
\nwc{\Backlund}{B\"{a}cklund}
\nwc{\Volpert}{Vol'pert}
\nwc{\cadlag}{c\`{a}dl\`{a}g}
\nwc{\BVloc}{BV_{loc}}
\nwc{\be}{\begin{equation}}
\nwc{\ee}{\end{equation}}
\nwc{\ba}{\begin{eqnarray}}
\nwc{\ea}{\end{eqnarray}}
\nwc{\la}{\label}
\nwc{\nn}{\nonumber}
\nwc{\Z}{\mathbb{Z}}
\nwc{\C}{\mathbb{C}}
\nwc{\E}{\mathbb{E}}
\nwc{\R}{\mathbb{R}}
\nwc{\N}{\mathbb{N}}
\nwc{\Mn}{\mathbb{M}_N}
\nwc{\prob}{\mathbb{P}}
\nwc{\Skor}{\mathbb{D}}
\nwc{\attr}{\mathcal{A}}
\nwc{\PP}{\mathcal{P}}
\nwc{\PPE}{\mathcal{P}(E)}
\nwc{\M}{\mathcal{M}}
\nwc{\Tt}{T^{(t)}}
\nwc{\Ut}{U^{(t)}}
\nwc{\Vt}{V^{(t)}}
\nwc{\Lt}{L^{(t)}}
\nwc{\tlambda}{\tilde{\lambda}}
\nwc{\tbar}{\bar{t}}
\nwc{\gtx}{g^{(t)}_x}
\nwc{\order}{\prec}
\nwc{\law}{\stackrel{\mathcal{L}}{\rightarrow}}
\nwc{\eqd}{\stackrel{\mathcal{L}}{=}}
\nwc{\vp}{\varphi}
\nwc{\Vp}{\Phi}
\nwc{\psilevy}{\Psi}
\nwc{\ve}{\varepsilon}
\nwc{\veps}{\varepsilon}
\nwc{\eps}{\ve}
\nwc{\betarsc}{\theta}
\nwc{\cl}{c\'{a}dl\'{a}g}
\nwc{\qref}[1]{(\ref{#1})}
\nwc{\D}{\partial}
\nwc{\Ebar}{\bar{E}}
\nwc{\mmt}{m}
\nwc{\dnto}{\downarrow}
\nwc{\fzero}{F_\rho} 
\nwc{\fone}{M_\rho}
\nwc{\ip}[1]{\langle #1 \rangle}
\nwc{\ipbig}[1]{\left\langle #1 \right\rangle}
\nwc{\Lip}{\mathop{\rm Lip}\nolimits}
\nwc{\Tmin}{T_{\min}}
\nwc{\Tmax}{T_{\max}}
\nwc{\Tgel}{T_{\rm gel}}
\nwc{\LL}{\mathcal{L}}
\nwc{\mudot}{\mu}
\nwc{\nudot}{\nu}
\nwc{\rme}{{\rm e}}
\nwc{\rmi}{{\rm i}}
\nwc{\nsup}{^{(n)}}
\nwc{\ksup}{^{(k)}}
\nwc{\jsup}{^{(j)}}
\nwc{\nksup}{^{(n_k)}}
\nwc{\inv}{^{-1}}
\nwc{\qxfac}{(1-\rme^{-qx})}
\nwc{\one}{\mathbf{1}}
\nwc{\psib}{\psi^{(b)}}
\nwc{\pss}{\psi_\#}
\nwc{\phss}{\Phi_\#}
\nwc{\asharp}{\gamma}
\nwc{\vfour}{\hat{v}}
\nwc{\Pss}{\Psi_\#}
\nwc{\pibar}{\Pi}
\nwc{\Ltail}{\bar{\Lambda}}
\nwc{\vb}{v^{(b)}}
\nwc{\spec}{\tilde{\beta}}
\nwc{\hv}{\hat{v}}
\nwc{\hvl}{\hat{v}_L}
\nwc{\lap}{{\,\cal L}}
\nwc{\argmin}{\mathrm{arg}^+\mathrm{min}}
\nwc{\argmax}{\mathrm{arg}^+\mathrm{max}}
\nwc{\gena}{\mathcal{A}}
\nwc{\genb}{\mathcal{B}}
\nwc{\genc}{\mathcal{C}}
\nwc{\testfn}{\varphi}
\nwc{\jumpex}{n_*}
\nwc{\jump}{n}
\nwc{\Jbar}{j}
\nwc{\Kbar}{k}
\nwc{\Ai}{\mathrm{Ai}}
\nwc{\Painleve}{Painlev\'{e}}
\nwc{\pot}{\mathbf{P}}
\nwc{\drift}{b}
\nwc{\den}{F}
\nwc{\velu}{V_u}
\nwc{\velv}{V_v}
\nwc{\genadag}{\gena^\dagger}
\nwc{\genbdag}{\genb^\dagger}
\nwc{\up}{u_+}
\nwc{\um}{u_-}
\nwc{\umh}{u_{-h}}
\nwc{\psip}{\psi_+}
\nwc{\psimh}{\psi_{-h}}
\nwc{\onefn}{$1$-point function}
\nwc{\twofn}{$2$-point function}
\nwc{\driftab}{\nu}
\nwc{\collab}{\sigma}
\nwc{\QQ}{\mathcal{Q}}
\nwc{\My}{\mathcal{M}}
\nwc{\Mf}{\mathcal{N}}
\nwc{\Markov}{\mathfrak{M}}
\nwc{\markov}{\mathfrak{m}}
\nwc{\cone}{\mathfrak{q}}
\nwc{\Stoch}{\mathfrak{Q}}
\nwc{\ad}{\mathrm{ad}}
\nwc{\Ad}{\mathrm{Ad}}
\nwc{\adR}{\ad^{(R)}}
\nwc{\diag}{\mathrm{diag}}
\nwc{\dd}{\mathfrak{d}}
\nwc{\gln}{\mathfrak{g}}
\nwc{\unn}{\mathfrak{u}}
\nwc{\lotri}{\mathfrak{l}}
\nwc{\GLM}{GL(m,\C)}
\nwc{\proj}{P}
\nwc{\projm}{\proj_\markov}
\nwc{\projd}{\proj_\dd}
\nwc{\projmp}{\proj_{\markov^\perp}}
\nwc{\projdp}{\proj_{\dd^\perp}}
\nwc{\Tr}{\mathrm{Tr}}
\nwc{\Id}{\mathrm{Id}}
\nwc{\loopalg}{L(\gln)}
\nwc{\loopp}{\loopalg_+}
\nwc{\loopm}{\loopalg_-}
\nwc{\Loopgroup}{\mathcal{G}}
\nwc{\son}{\mathfrak{so}(N)}
\nwc{\schwartz}{\mathcal{S}}
\nwc{\id}{I}
\nwc{\Res}{\mathrm{Res}}
\nwc{\group}{g}
\nwc{\length}{L}
\nwc{\action}{S}
\renewcommand{\Re}{\mathop{\rm Re}\nolimits}

\theoremstyle{definition}
\newtheorem{defn}[theorem]{Definition} 
\newtheorem{remark}[theorem]{Remark}

\theoremstyle{remark}
 
\numberwithin{equation}{section}
\numberwithin{figure}{section}
\begin{document}

\title{Complete integrability of shock clustering and Burgers turbulence\\} 
\author{Govind Menon\textsuperscript{1}}

\date{\today}

\maketitle

\begin{abstract}
We consider scalar conservation laws with convex flux and random initial data. 
The Hopf-Lax formula induces a deterministic evolution of the law of the initial data. In a recent article, we derived a kinetic theory and Lax equations to describe the evolution of the law under the assumption that the initial data is a spectrally negative Markov process. Here we show that: (i) the Lax equations are Hamiltonian and describe a principle of least action on the Markov group that is in analogy with geodesic flow on $SO(N)$; (ii) the Lax equations are completely integrable and linearized via
a loop-group factorization of operators; (iii) the associated zero-curvature equations can be solved via inverse scattering. Our results are rigorous for $N$-dimensional approximations of the Lax equations, and yield formulas for the limit $N \to \infty$. The main observation is that the Lax equations are a $N \to \infty$ limit of a  Markovian variant of the $N$-wave model. This allows us to introduce a variety of methods from the theory of integrable systems.

\end{abstract}



\medskip
\noindent
\footnotetext[1]
{Division of Applied Mathematics, Box F, Brown University, Providence, RI 02912.
Email: menon@dam.brown.edu}

\section{Introduction}
\label{sec:intro}
\subsection{Turbulence and flows of probability measures}
\label{subsec:turb}
A fundamental problem in the statistical theory of turbulence is to construct random incompressible velocity fields that model isotropic homogeneous turbulence. Such random fields must be supported on weak solutions to the Euler equations that dissipate kinetic energy in accordance with the criterion of Kolmogorov and Onsager~\cite{CET,Eyink,Kolm,Onsager}. This problem is currently out of reach. Much of our understanding is based instead on vastly simplified models. One such model, proposed by Burgers, is to understand the statistics of the Cole-Hopf (or entropy) solution to Burgers equation
\be
\label{eq:burgers}
\partial_t u + \partial_x \left(\frac{u^2}{2}\right) =0, \quad x \in \R, t>0,
\ee
with random initial data $u_0$ such as white noise~\cite{Burgers,She}.  A closely related problem is to consider \qref{eq:burgers} with random forcing, and to understand the associated equilibrium measure. While we focus on the unforced equation in this article, some of our methods also apply to \qref{eq:burgers} with random forcing~\cite{MS1}. This class of problems is called Burgers-KPZ turbulence. An explanation of its role in statistical hydrodynamics may be found in~\cite{E-Sinai}. 

Random initial data also arise in applications unrelated to turbulence. For example, the coarsening of domains in the kinetics of phase transitions is modeled by solutions to the Allen-Cahn and Cahn-Hilliard equations that emerge from disorder~\cite{Voorhees}. In such problems, the equations of continuum physics induce an evolution of the law of the initial data. The first clear formulation of an evolution equation for the law of solutions seems to be due to Hopf~\cite{Hopf-stat}.  In this article, we show that this problem is suprisingly rich, even in the setting of perhaps the simplest nonlinear equations. 

We consider the  scalar conservation law 
\be
\label{eq:sc1}
\partial_t u + \partial_x f(u) =0, \quad x \in \R, t>0, \quad u(x,0)=u_0(x),
\ee
with a strictly convex, $C^1$ flux $f$. The Hopf-Lax formula defines a unique entropy solution to \qref{eq:sc1}.  When the initial data $u_0$ is random, the Hopf-Lax formula induces a deterministic evolution of the law of $u_0$.  Our main contribution is to show that if $u_0$ is a Markov process in $x$ with only downward jumps or a limit of such processes, then the {\em evolution of its law  is  completely integrable\/}. Our results include Burgers model, but the assumption $f(u)=u^2/2$ is not necessary.  Neither do we require a special choice of initial condition such as white noise or Brownian motion (though these yield important exact solutions). For all convex, $C^1$ $f$ and a broad class of random initial data the evolution of the law of $u(\cdot, t)$ is given by kinetic equations for shock clustering that are a continuum limit of a Markovian variant of the $N$-wave model.  The 
$N$-wave model generalizes the three-wave model of Manakov and Zakharov in nonlinear optics~\cite{Zakharov-Manakov} and is well-known to be completely integrable~\cite[p.55]{Ablowitz},~\cite[III.4]{Zakharov}. We show that 
this model also underlies Burgers turbulence and shock clustering for 
~\qref{eq:sc1} with arbitrary convex $f$. We stress that it is the evolution of the  {\em law\/} of $u(\cdot,t)$ that is integrable, and this is completely distinct from the integrability of \qref{eq:burgers}.

This is surprising enough, but more is true. Our kinetic equations sit at a rich juncture of  problems: (i) geodesic flows on Lie groups~\cite{Manakov}; (ii) the integrable systems of the 19th century, as recast by Moser~\cite{Moser}; (iii) the completely integrable systems of random matrix theory discovered by the Kyoto school~\cite{JMMS}; (iv) integrable hierarchies on groups~\cite{TU1}; and (v) asymptotic problems in representation theory~\cite{Kerov}. This reveals a close and unexpected relation between the theory of Markov processes, kinetic theory and integrable systems.

This confluence of ideas is quite bewildering, and a full explanation for their role in what should be a purely probabilistic problem still eludes us. Some heuristic explanation is perhaps the following: At its heart, complete integrability is an explicit understanding of the hidden symmetries of a Hamiltonian system. On the other hand, understanding continuous deformations of the law of a stochastic process is a basic problem in probability theory (e.g., the Girsanov theorem may be viewed in this light).  The deeper principles here seem to be  that: (i) the space of Feller processes with bounded variation on the Skorokhod space can be given a natural symplectic structure; (ii) for every convex, $C^1$ flux $f$ \qref{eq:sc1} induces a Hamiltonian flow on this space with respect to this symplectic structure; (iii) these flows  commute for distinct $f$. A precise formulation of these ideas is subtle, and we have been unable to develop this viewpoint completely, even though we obtain partial results. In order to state precisely what we prove, and what is mere conjecture, we first review some recent work.

\subsection{Lax equations for shock clustering with Markov data}
We always assume that random initial data $u_0$ for \qref{eq:sc1} is a Markov process in $x$ with only downward jumps (a {\em spectrally negative Markov process}) or a limit of such processes. This assumption is motivated by two considerations. First, in order to obtain a detailed understanding of the evolution of the law of $u_0$ under \qref{eq:sc1}, it is necessary to work with a class of well-understood random processes on the line and it is natural to choose Markov processes. Second, it is a surprising fact that for every $t>0$ the solution to Burgers equation with white noise is a stationary, spectrally negative Markov process in $x$. The Markov property of this solution  was assumed by Burgers, and first proved (in a completely different context) by Groeneboom~\cite{Groeneboom}. The importance of the Markov property in the context of Burgers turbulence was first noted by Avallaneda and E~\cite{AE}. 

In a recent article, we proved the  following {\em closure theorem} for the entropy solution to \qref{eq:sc1}:  Assume the initial data $u_0(x)$ is a  spectrally negative strong Markov process in $x$. Then for every $t>0$ the entropy solution to \qref{eq:sc1} remains a spectrally negative Markov process in $x$~\cite[Thms.2,3]{MS1}. This shows that the entropy solution to \qref{eq:sc1} leaves this class of stochastic processes invariant. It is not necessary to assume that $f(u)=u^2/2$, only that $f$ is convex and $C^1$. There is also no need to assume that $u_0$ is stationary in $x$.  Under an additional assumption of regularity (preservation of the Feller property), the closure theorem forms the basis for a kinetic theory of shock clustering as follows. 

Feller processes are characterized by their generators. The simplest such characterization is the \LK\/ formula for \Levy\/ processes. A general  characterization, attributed to Courr\'{e}ge in~\cite[Thm 3.5.3]{Applebaum} builds on the \LK\/ formula. Assume $u(x)$, $x \in \R$ is a stationary, Feller process. Then its generator $\gena$ is an integro-differential operator that acts on $C_c^\infty$ test functions in its domain as follows:
\ba
\label{eq:couregge}
\lefteqn{\gena \testfn(u) = a(u)\testfn''(u) + \drift(u)\testfn'(u) + c(u)\testfn(u)}
\\
&&
\nn
 + \int_{\R \backslash\{u\}} \left( \testfn(v)-\testfn(u) - \psi(u,v) \testfn'(u)\right) \jump(u,dv).
\ea
Here $a$,$\drift$,$c$ are functions on the line satisfying certain continuity criterion, $n(u,dv)$ is a \Levy\/ measure that describes the jumps of the process, and $\psi(u,v)$ is a local unit (in the simplest situation, we have $\psi(u,v)=v-u$). The precise assumptions are stated in~\cite[Thm. 3.5.3]{Applebaum}. There is an intimate relation between this characterization and the sample paths of the Feller process: $a$ describes the diffusion of the Feller process, so that $a \geq 0$; $c$ describes killing, so that $c \leq 0$; $\drift$ describes the drift, and its sign is not restricted.

Our basic idea, following~\cite{CD}, is to study the shock statistics through an evolution equation for the generator. The general formula \qref{eq:couregge} simplifies because for any $t>0$, the entropy solution to \qref{eq:sc1} has bounded variation, only downward jumps, and no killing. Thus $a$ and $c$ vanish, and the support of $n(u,dv)$ is $(-\infty,u)$. For fixed $t>0$, if $u(x,t)$, $x \in \R$ is a stationary  Feller process, its generator $\gena(t)$ is an integro-differential operator of the form
\be
\label{eq:gena1}
\gena(t)\testfn(u) = \drift(u,t) \testfn'(u) + \int_{-\infty}^u \left( \testfn(v)-\testfn(u)\right) \jump(u,dv,t). 
\ee
One of the main results in ~\cite{MS1} is that $\gena(t)$ satisfies the Lax equation
\be
\label{eq:kinetic}
\partial_t \gena = [\gena, \genb].
\ee
Here $[\gena,\genb] =\gena\genb-\genb \gena$ denotes the Lie bracket, and 
the operator $\genb$ is defined by its action on test functions as follows:
\be
\label{eq:genb}
\genb \testfn (u) = -f'(u)\drift(u,t) \testfn'(u) -\int_{-\infty}^u [f]_{u,v}\left( \testfn(v)-\testfn(u)\right) \jump(u,dv,t).
\ee
$[f]_{u,v}$ is abbreviated notation for the Rankine-Hugoniot speed of a shock connecting states $u$ and $v$.
\be
\label{eq:RH}
[f]_{u,v} := \frac{f(v)-f(u)}{v-u}.
\ee
We do not need to assume that the process $u(\cdot, t)$ is stationary in $x$. Non-stationary data arise when we wish to model the spread of `turbulent bursts' (i.e. localized data $A_0(x)$) or  Riemann data (e.g one-sided initial conditions as in ~\cite{Valageas}). These situations are described by the zero-curvature equation
\be
\label{eq:zc}
\partial_t \gena -\partial_x \genb = [\gena, \genb].
\ee

It requires considerable insight to realize that this approach is fruitful, and our work was greatly inspired by Duchon and co-workers~\cite{Duchon0,Duchon1,CD}. Several open questions remain: In particular, it is still necessary to justify the connection between \qref{eq:sc1} with random data and \qref{eq:zc} in full generality~\cite[\S 1.6]{MS1}. Our earlier work only showed that the entropy solution preserves the strong Markov property, whereas we really need to show that it preserves the Feller property. This requires a well-posedness theory and uniform estimates for \qref{eq:zc}. One of our goals in this article is to lay a foundation for the analysis of \qref{eq:zc} that also allows us to address this question.

\subsection{Kinetic equations}
To convince the reader of the merit of this approach, let us briefly explain how it describes the evolution of shock statistics.  Since the generator $\gena$ is characterized by the drift, $\drift$, and jump measure,  $\jump$, the Lax equation can be expanded using \qref{eq:gena1} and \qref{eq:genb} to yield evolution equations for $\drift$ and $\jump$. The drift satisfies the simple differential equation
\be
\label{eq:evolb}
\partial_t \drift(u,t) = -f''(u)\drift^2(u,t).
\ee
The evolution of $\jump$ is more interesting, since the shock statistics evolve by decay of rarefaction waves and growth in binary collisions. To state the evolution equations, we assume for simplicity that the jump measure has a density, say $\jump(u,dv,t)= \jump(u,v,t) \, dv$. Then $\jump$ satisfies a kinetic equation of Vlasov-Boltzmann type
\ba
\label{eq:kin26}
\lefteqn{ \partial_t \jump(u,v,t) + \partial_u\left(\jump \velu (u,v,t)\right)+ \partial_v \left( \jump\velv(u,v,t) \right)} \\
\nonumber 
&&
= Q(\jump,\jump) + \jump \left( \left([f]_{u,v}-f'(u)\right)\partial_u \drift - \drift f''(u)\right).
\ea
Here the drift velocities $\velu$ and $\velv$ are defined by
\be
\label{eq:kin9}
\velu(u,v,t) = \left( [f]_{u,v} -f'(u)\right) \drift(u,t), \quad \velv(u,v,t) = \left([f]_{u,v}-f'(v)\right) \drift(v,t),
\ee
and the collision kernel $Q$ counts growth and loss in binary clustering
\ba
\nonumber
\lefteqn{ Q(\jump,\jump)(u,v,t) = \int_v^u \left([f]_{u,w}-[f]_{w,v}\right)\jump(u,w,t)\jump(w,v,t)\, dw}\\
\nonumber
&& - \int_{-\infty}^v \left([f]_{u,v}-[f]_{v,w}\right) \jump(u,v,t)\jump(v,w,t)\, dw \\
\label{eq:kin24}
 && -\int_{-\infty}^u \left([f]_{u,w}-[f]_{u,v}\right)\jump(u,v,t)\jump(u,w,t)\, dw. 
\ea
The distinction with earlier work on kinetics of shock clustering~\cite{EvE,Kida} is summarized in~\cite{MS1}.

\subsection{Exact solutions}
Despite their formidable appearance, these kinetic equations admit surprising exact solutions for Burgers equation. The first class of solutions correspond to spectrally negative \Levy\/ processes (e.g. when $u_0$ is a Brownian motion as in~\cite{Sinai}).  We then have  
$\gena e^{qy}= \psi(q,t) e^{qy}$ and $[\gena,\genb]e^{qy} = -\psi \partial_q \psi(q,t)$, where $\psi(q,t)$ denotes the Laplace exponent of the process $u(x,t)-u(0,t)$, $x \geq 0$, for fixed $t>0$. Then the Lax equation \qref{eq:kinetic} yields
\be
\label{eq:duchon}
\partial_t \psi + \psi \partial_q \psi =0, \quad t >0, q >0.
\ee
The beautiful fact that the Laplace exponent itself evolves by Burgers equation was discovered by Carraro and Duchon~\cite{Duchon0}, and made rigorous by Bertoin~\cite{B_burgers}. It is reminescent of the inverse scattering method in integrable systems.  In addition,~\qref{eq:kin26} reduces to Smoluchowski's coagulation equations with additive kernel. It has been known for some time that Smoluchowski's equation can be solved explicitly by the Laplace transform~\cite{Golovin}. But this solution takes on new meaning when we recognize that it describes exactly the clustering of shocks in Burgers equations with \Levy\/ process data. This connection is the basis for several deeper results connecting stochastic coalescence and Burgers turbulence~\cite{B_icm}.

Another remarkable solution to \qref{eq:kinetic} corresponds to the shock statistics in Burgers equation with white noise initial data~\cite{Frachebourg,Groeneboom}. This corresponds to a self-similar solution to \qref{eq:evolb} and \qref{eq:kin26} of the form 
\be
\label{eq:burg_wn}
b(u,t) = \frac{1}{t}, \quad \jump(u,v,t) = \frac{1}{t^{1/3}}\jumpex(ut^{1/3},vt^{1/3}).
\ee
The jump density $\jump_*$ of the integral operator is given explicitly as follows:
\be
\label{eq:groen2}
\jumpex(u,v) = \frac{J(v)}{J(u)}K(u-v),  \quad u >v, 
\ee
and vanishes if $u \geq v$. Here $J$ and $K$ are positive functions defined on the line and positive half-line respectively, whose Laplace transforms
\be
\label{eq:groen3}
\Jbar(q) =\int_{-\infty}^\infty e^{-qy} J(y)\,dy, \quad \Kbar(q) =\int_0^\infty e^{-qy} K(y)\, dy,
\ee
are meromorphic functions on $\C$ given by 
\be
\label{eq:groen4}
\Jbar(q) =\frac{1}{\Ai(q)}, \quad \Kbar(q)= -2 \frac{d^2}{dq^2}\log \Ai(q).
\ee
$\Ai$ denotes the Airy function~\cite[10.4]{AS}. This generator was computed by Groeneboom (but not as a solution to \qref{eq:kinetic}!)~\cite{Groeneboom}. When written this way, the formula for $\Kbar$ is reminescent of determinantal formulas in soliton theory.

\subsection{Discrete Lax equations and the Markov $N$-wave model}
The most intriguing aspect of~\cite{MS1} is that the kinetic theory of shock clustering has many features reminiscent of completely integrable systems. These include the formulation as a Lax pair, explicit exact solutions as above, a Painlev\'{e} property for the self-similar solution, and links with random matrix theory (see~\cite[\S 1.5.1]{MS1}). Our main contribution here is to understand the origin of these coincidences. 

In order to show that \qref{eq:kinetic} is completely integrable, we must show that it defines a Hamiltonian system and construct infinitely many commuting integrals. In addition, what is required is an explicit solution  via inverse scattering or a Riemann-Hilbert problem. In order to address these questions, we first study  exact finite-dimensional discretizations of \qref{eq:kinetic} (this terminology is explained below). Here the definition of Hamiltonian structure and complete integrability are unambiguous, and it has to be shown that \qref{eq:kinetic} actually has this structure. This lays a foundation for the analysis of \qref{eq:kinetic} and \qref{eq:zc}.

We discretize the system as follows. Fix a positive integer $N$ and discrete velocities $-\infty < u_1 < u_2 < \ldots < u_N < \infty$. We consider a continuous $x$ Markov processes that takes the values $u_k$, $1 \leq k \leq N$. The sample paths of this Markov process are piecewise constant paths. Let $\Mn$ denote the space of $N\times N$ matrices.  The generator $A \in \Mn$ satisfies
\be
\label{eq:Markov1} A_{ij} \geq 0, i \neq j, \quad \sum_{j=1}^N A_{ij}=0.
\ee
More formally, $A$ corresponds to an operator $\gena$ with drift $b \equiv 0$ and a jump measure $\jump(u_j,\,dv) = \sum_{k \neq j}A_{jk}\delta_{u_k}(dz)$. The jump measure vanishes if $u$ is not one of the discrete states $u_j$. In order to define the discretization $B$ of $\genb$ we introduce the symmetric matrix
\be
\label{eq:Markov2} F_{ij} = -\frac{f(u_i)-f(u_j)}{u_i-u_j}, i \neq j, \quad F_{ii} = - f'(u_i).
\ee
Here $f$ is the flux in the conservation law \qref{eq:sc1}, so that $-F_{ij}$ is the Rankine-Hugoniot speed of the shock connecting states $u_i$ and $u_j$. We then set
\be
\label{eq:Markov3} 
B_{ij} =F_{ij} A_{ij}, \;i \neq j,\quad B_{ii} = -\sum_{j \neq i} B_{ij}.
\ee
The discrete zero-curvature equation  is
\be
\label{eq:discrete-zc}
\partial_t{A} -\partial_x B =[A,B],
\ee
and the discrete Lax equation is
\be
\label{eq:discrete}
\dot{A}=[A,B].
\ee

Let us briefly comment on why this discretization is natural.  Initial data to \qref{eq:discrete-zc} correspond to random initial data $u_0$ to \qref{eq:sc1}. In particular, initial data to \qref{eq:discrete-zc} that are lower triangular corresponds to $u_0$ that are spectrally negative Markov processes with a finite number of velocities. For such initial data, the Feller property is preserved in time and the evolution of statistics of random initial data by the conservation law \qref{eq:sc1} is described exactly by \qref{eq:discrete-zc}. Thus, solutions to \qref{eq:discrete-zc} that are lower triangular and generators of Markov processes correspond to exact solutions to \qref{eq:zc}. This is what we mean when we say that the discretization is exact. 

%It is necessary here to work with \qref{eq:discrete-zc} instead of \qref{eq:discrete} since stationary processes (in $x$) with a finite number of states and only downward jumps are trivial. 

When we consider \qref{eq:discrete}, complete integrability immediately becomes plausible. Indeed, ordinary differential equations of the form \qref{eq:discrete} arise in basic examples in the theory of integrable systems: these include geodesic flows on $SO(N)$~\cite{Manakov}, the integrable flows of Neumann and Jacobi~\cite{Moser}, and the integrable flows of random matrix theory~\cite{JMMS}. Moreover, the zero-curvature equations \qref{eq:discrete-zc} are very similar to the $N$-wave model~\cite{Ablowitz,Zakharov} with an important difference. For the $N$-wave model we typically assume $A \in \unn$, the algebra of the unitary group.
(Perhaps the only study  where $A$ is {\em not\/} assumed to lie in $\unn$ is~\cite{Ab-Hab}).  But $A \in \unn$ is incompatible with the Markov property.  Motivated by these considerations, we call \qref{eq:discrete-zc} the Markov $N$-wave model, or $mN$-wave model for short.

\subsection{Statement of results}
A complete study of \qref{eq:kinetic} and \qref{eq:zc} requires a combination of methods from the theory of Markov processes, integrable systems and spectral theory.  Our approach is to prove complete results in the finite-dimensional setting in this article, and to study the limit $N \to \infty$ in a sequel. Once one has  recognized the structure of the problem in the discrete setting, the proofs only rely on well-established techniques from integrable systems. (Of course, the main difficulty in our work was to recognize this structure!).  Our main results are:
\begin{enumerate}
\item  \qref{eq:discrete} is Hamiltonian  and is associated to a principle of least action (see Section~\ref{sec:hamilton} and Theorem~\ref{thm:hamilton}). 
\item  \qref{eq:discrete} is completely integrable  and linearized by a loop-group factorization (see Section~\ref{sec:aks} and Theorem~\ref{thm:ci}).
\item  an inverse scattering theory  and well-posedness for \qref{eq:discrete-zc} (see Section~\ref{sec:scatter}, in particular Theorem~\ref{thm:positivity}).
\end{enumerate}
Though obviously related, the method of solution for \qref{eq:kinetic} and \qref{eq:zc} turn out to be quite distinct. The Lax equation~\qref{eq:kinetic} can be studied through an operator factorization problem, while the zero-curvature equation \qref{eq:zc} can be attacked by inverse scattering.

Let us briefly explain how these results are proved. Hamiltonian structures and complete integrability are deeply linked to group actions. In seeking an algebraic, but probabilistically natural, approach to \qref{eq:kinetic}, we find the following simple structure. We observed in~\cite[\S 2.7]{MS1} that the space  $\markov_\infty$ of integro-differential operators $\genc$ of the form
\be
\label{eq:m_infinity}
\genc \testfn(u) = \beta (u) \testfn'(u) + \int_{\R} \left( \testfn(v)-\testfn(u)\right) \nu(u,v) \, dv, \quad \beta \in C_c^\infty(\R), \nu \in 
C_c^\infty(\R^2),
\ee
formally constitutes a Lie algebra. Of course, it is not clear that these operators generate a infinite-dimensional Lie group. But we may use this insight at the discrete level quite easily. The set of generators of Markov processes defined by \qref{eq:Markov1} forms a cone $\cone_N \subset \Mn$. Each element of $\cone_N$ generates a one-parameter Markov semigroup. Let us define 
\be
\label{eq:semigroup}
\Stoch_N = \{g \in GL(N, \R)| \quad g = e^{xA}, \; A \in \cone_N, \;x \geq 0 \}.
\ee
$\Stoch_N$ can be naturally embedded in a Lie group as follows. Define the {\em Markov algebra\/}
\be
\label{eq:m_algebra}
\markov_N= \{ A \in \Mn | \quad \sum_{j=1}^N A_{ij}=0, \quad 1 \leq i \leq N\}. 
\ee
$\markov_N$ is a Lie algebra. It generates the {\em Markov group\/}
\be
\label{eq:m_group}
\Markov_N= \{ g \in GL(N,\R) | \quad g = e^{xA}, \quad A \in \markov_N, \; x \in \R. \},
\ee
Since $N$ will be fixed in our results, we will mostly suppress the subscript $N$ in what follows. Clearly $\cone \subset \markov$, and $\Stoch \subset \Markov$. Some basic properties of $\markov$ and $\Markov$ may be found in~\cite{Johnson}.

The first result is that \qref{eq:discrete} is a Hamiltonian flow on $\markov$. In addition, when $f$ is convex and $A$ lower triangular, this Hamiltonian flow leaves $\cone$ invariant. These results are seen as follows.
Co-adjoint orbits of Lie groups carry a natural (Kirillov-Kostant) symplectic structure. In addition, Lie algebras that admit direct sum decompositions into subalgebras carry more than one symplectic structure. This idea has been formalized by the notion of an $r$-matrix~\cite{Semenov}. The Lie algebra $\gln = gl(N,\R)=\Mn$ admits a natural splitting induced by $\markov$ (see~\qref{eq:decomp1} below). We show that \qref{eq:discrete} is a Hamiltonian system on the algebra $\gln$ with a Lie-Poisson bracket induced by this splitting. This is stated precisely in Theorem~\ref{thm:hamilton} below.  Once we have established the Hamiltonian structure of \qref{eq:discrete} we treat the probabilistically important case ($A \in \cone$) in Theorem~\ref{thm:convex}. The proof shows clearly the role of convexity of $f$ and spectral negativity at the level of the Lax equations. 

We then formulate an associated principle of least action for \qref{eq:discrete}. This is in precise analogy with geodesic flow on $SO(N)$ and in particular, with Euler's equation for a free rigid body. When $F_{ij}$ is positive, the role of the flux $f$ is to define a (degenerate) metric on $\Markov$. More generally, $f$ defines a quadratic action through the multiplier $F_{ij}$. This principle of least action describes the evolution of a probability measure on path space through \qref{eq:discrete}, and is completely distinct from the usual principle of least action for the Hopf-Lax solution to \qref{eq:sc1}. 

Complete integrability is based on an elegant observation of Manakov~\cite{Manakov}.  Define the diagonal matrices
\be
\label{eq:diagonals}
\My = \diag(u_1, \ldots, u_N), \quad \Mf=\diag(f(u_1), \ldots, f(u_N)),
\ee
and  observe that $A$ and $B$ are related through the algebraic relation
\be
\label{eq:manakov}
[A, \Mf] -[\My,B]=0.
\ee
This allows us to introduce a spectral parameter $z \in \C$ in  \qref{eq:discrete} and embed the flow in a loop-algebra 
\be
\label{eq:manakov2}
\frac{d}{dt}(A-z\My) = [A-z\My, B+ z\Mf], \quad z \in \C.
\ee
Thus, the spectral curve $\{(z,\lambda) \in \C^2 | \det(A - z \My - \lambda \id)=0 \}$ is invariant, and its coefficients are integrals. Rather than verify explicitly that these integrals are in involution as in~\cite{Manakov}, it is simpler to apply the Adler-Kostant-Symes (AKS) theorem to show that \qref{eq:manakov2} is completely integrable. Again we need to find a suitable $r$-matrix, now of the loop algebra.
This is stated precisely in Theorem~\ref{thm:ci}.  The matrix factorization in the AKS theorem also yields a Riemann-Hilbert problem. Once this is formulated for finite $N$, it also yields a limiting factorization problem for the Lax equation \qref{eq:kinetic}.

Scattering and inverse scattering for a class of integrable systems including the $N$-wave model on $\unn$ was established by Beals and Coifman~\cite{BC1,BC2} and formulated in a general Lie algebraic setting by Beals and Sattinger~\cite{BS1}. 
This theory does not directly apply to \qref{eq:discrete-zc} since $\markov$ is not semi-simple. We modify~\cite{BC1,BC2} to obtain a scattering and inverse scattering theory for \qref{eq:discrete-zc}. Among other results, we obtain a hierarchy of  integrable flows and global well-posedness theorems for \qref{eq:discrete-zc} including the probabilistically natural case. These results are contained in Section~\ref{sec:scatter}. In addition to these rigorous results, the method also yields a formal solution procedure for \qref{eq:zc} via inverse scattering.

In the limit $N \to \infty$, $\gena$ is an integro-differential operator of the form \qref{eq:gena1}.  Subtle problems arise in the approximation: under suitable assumptions, an operator of the form \qref{eq:couregge} can be approximated by operators in $\markov_N$. However, operators of the general form \qref{eq:couregge} are not closed under the commutator! In addition, while operators of the form \qref{eq:m_infinity} form a Lie algebra, it is not clear that they correspond to an infinite-dimensional Lie group. Nevertheless, our work yields a  formal understanding of \qref{eq:kinetic} and \qref{eq:zc}, and thus \qref{eq:sc1} with spectrally negative Markov data. Theorem~\ref{thm:hamilton} suggests that \qref{eq:kinetic} is Hamiltonian with an associated principle of least action on the semigroup of Markov operators. Theorem~\ref{thm:ci} suggests that \qref{eq:kinetic} is  completely integrable, and yields an operator factorization problem for \qref{eq:kinetic}. Finally, the inverse scattering problem for \qref{eq:zc} extends formally to unbounded operators with little change. We hope that these natural conjectures will stimulate rigorous results of full generality.

\section{Hamiltonian structure}
\label{sec:hamilton}
In this section, we show that the Lax equation \qref{eq:discrete} defines a Hamiltonian system. As is well-known,  co-adjoint orbits of a Lie group carry a natural symplectic structure~\cite{AdvM-book,Arnold-Khesin}.  If $\gln$ denotes a finite-dimensional Lie algebra, and $H$ a smooth Hamiltonian defined on the dual space $\gln^*$, then Hamilton's equations may be rewritten as Kirillov's equation
\be
\label{eq:kirillov}
\dot{\alpha} = \ad^*_{dH(\alpha)}(\alpha), \quad \alpha \in \gln^*.
\ee
Our main observation is that the Lax equation \qref{eq:discrete} takes this form on the algebra $\gln$ with a probabilistically natural bracket defined in \qref{eq:twisted} below. 

\subsection{Algebraic preliminaries}
\label{sec:algebra}
Let $\gln =gl(N,\R)$ denote the Lie algebra of real, $N\times N$ matrices equipped with the bracket $[A,B]=AB-BA$. We have already defined the subspace $\markov \subset \gln$ in \qref{eq:m_algebra}. It is easily checked that $\markov$ is a subalgebra of $\gln$. Let $\dd$ denote the subspace of diagonal matrices. Since diagonal matrices commute, $\dd$ is trivially a subalgebra. Its importance here lies in the fact that $\gln$ admits a direct sum (vector space) decomposition 
\be
\label{eq:decomp1}
\gln = \markov \oplus \dd.
\ee
The decomposition \qref{eq:decomp1} is obtained as follows. Recall that $e = (1, 1, \ldots, 1)^T$. Given $A \in \gln$, we define the projections
\be
\label{eq:decomp2}
\projm A = A - \diag(Ae), \qquad \projd A = \diag(Ae).
\ee
Then $\projm^2 =\projm$, $\projd^2=\projd$, $\projm A \in \markov$ and $\projd A \in \dd$, and every matrix $A \in \gln$ may be written as
\be
\label{eq:decomp3}
A = \projm A + \projd A. 
\ee
Associated to a splitting such as \qref{eq:decomp1} is an $r$-matrix~\cite{Semenov}. 
This allows us to introduce a new Lie bracket on $\gln$. If $A,B \in \gln$, then we define a new $\ad$-action
\be
\label{eq:twisted}
\ad_B^{r} A= [B,A]_r= [\projm B,\projm A] - [\projd B, \projd A] = [\projm B, \projm A].
\ee
The last equality holds because diagonal matrices commute. $\gln$ remains a Lie algebra with the new bracket $[\cdot,\cdot]_r$.

We identify $\gln$ with its dual space $\gln^*$ through the non-degenerate $\Ad$-invariant pairing
\be
\label{eq:ad-pair}
(\alpha,A) =\Tr(\alpha A), \quad \alpha \in \gln^*, A \in \gln.
\ee 
The dual spaces $\markov^*$ and $\dd^*$ are naturally identified with the orthogonal complements $\dd^\perp$ and $\markov^\perp$ under $(\cdot,\cdot)$. It is easy to compute
\be
\label{eq:complements}
\markov^* \cong \dd^\perp = \{ \alpha | \, \diag(\alpha) =0 \}.
\ee
Since $Ae=0$ for every $A \in \markov$, and the dimension of $\markov$ is $N^2-N$ we also find
\be
\label{eq:complements2}
\dd^* \cong \markov^\perp = \{ \alpha | \, \alpha = \sum_{j=1}^N c_j E_j=0  \},
\ee
where $c_j \in \R$, $1 \leq j \leq N$ and $E_j$ is the matrix obtained from the zero matrix by replacing the $j$-th column with $e$. The projections $\projm$ and $\projd$ induce dual projections $\projmp$ and $\projdp$ in the natural manner: for $A \in \gln$ and $\alpha \in \gln^*$
\be
\label{eq:comp_proj}
(\projdp \alpha, A ) = (\alpha, \projm A), \quad  (\projmp \alpha, A ) = (\alpha, \projd A).
\ee

We now compute the $\ad^*$ action with the bracket  \qref{eq:twisted} and  the non-degenerate pairing \qref{eq:ad-pair}. For every
$A,B \in \gln$ and $\alpha \in \gln^*$ we have
\ba
\nonumber
\lefteqn{ \ad^{* r}_{B} \alpha (A)  = (\alpha, \ad^{r}_B A) =(\alpha, [\projm B, \projm A]) }
\\
\nn
&&
=  \Tr \left( \alpha [\projm B, \projm A] \right) = \Tr \left( [\alpha, \projm B] \projm A \right)
\\
\nonumber
&& 
 = \left( [\alpha,\projm B], \projm A \right) =
( \projdp  [\alpha, \projm B], A).
\ea
Since this holds for every $A \in \gln$,  we find
\be
\label{eq:baby_kir}
\ad^{*r}_{B} \alpha = \projdp [ \alpha, \projm B].
\ee

\subsection{Quadratic Hamiltonians and Kirillov's equation}
The calculation so far has been purely algebraic and we have been careful to distinguish $\alpha \in \gln^*$ from  $A \in \gln$. Since we have now computed the $\ad^{*r}$ action and $\gln$ has been identified with $\gln^*$ via $(,)$ we may drop this notation. We assume $H: \gln \to \R$ is $C^1$, replace $\alpha$ by $A$ and $B$ by $dH(A)$ in \qref{eq:baby_kir} to obtain Kirillov's equation with the bracket \qref{eq:twisted}
\be
\label{eq:kir_proj}
\dot{A} = \ad_{dH(A)}^{*r} A = \projdp [ A, \projm dH(A)]
\ee
These calculations have made precise the Hamiltonian structure. We now show that the Lax equation~\qref{eq:discrete} is Hamiltonian with this symplectic structure. 

\begin{theorem}
\label{thm:hamilton}
Assume given a real symmetric matrix $F$, and let $F \circ A$ denote the Hadamard product $(F \circ A)_{ij} = F_{ij} A_{ij}$. Define the quadratic Hamiltonian $H: \gln \to \R$ 
\be
\label{eq:quad_hamilton}
H(A) = \frac{1}{2} \Tr (A F\circ A) =\frac{1}{2} \sum_{i,j=1}^N F_{ij}A_{ij}A_{ji}.
\ee
\noindent
(a) The associated Hamiltonian vector field on $\gln$ with the bracket \qref{eq:twisted} is  
\be
\label{eq:kir_quad}
\dot{A}=[A, \projm F \circ A].
\ee
\noindent
(b) $\markov^\perp$, $\dd^\perp$, $\markov$, and $\dd$  are invariant under \qref{eq:kir_quad}.

\noindent
(c) When $F$ is defined by \qref{eq:Markov2}, the vector field \qref{eq:kir_quad} is identical to \qref{eq:discrete}.
\end{theorem}
\begin{proof}
(a) Kirillov's equation on $\gln^*$ with the bracket \qref{eq:twisted} takes the form \qref{eq:kir_proj}. We only need to show that
\be
\label{eq:kir_proj2}
\projdp [ A, \projm dH(A)] = [A, \projm dH(A)] = [A, \projm F\circ A]. 
\ee
The identity \qref{eq:kir_proj2} is seen as follows. By \qref{eq:quad_hamilton}, $dH(A)=F\circ A$ and 
\be
[A, dH(A)]_{ij} = \sum_{k=1}^N\left(F_{kj} -F_{ik}\right) A_{ik}A_{kj}.
\ee
Since $F$ is symmetric, $[A, dH(A)]$ vanishes on the diagonal. In addition, since $\projd dH(A)$ is a diagonal matrix, $[A, \projd dH(A)]$ also vanishes on the diagonal. Thus, by \qref{eq:complements} 
\[ \projdp [ A, dH(A)] = [A, dH(A)], \quad \projdp [ A, \projd dH(A)] = [A, \projd dH(A)], \]
which implies \qref{eq:kir_proj2}.

(b) Every vector field of the form \qref{eq:kir_proj} vanishes on $\markov^\perp$. Thus, all such vector fields (not just quadratic Hamiltonians) leave $\dd^\perp$ and $\markov^\perp$ invariant.  If $A \in \dd$, then $\projm dH(A)$ vanishes. Thus, $\dd$ is invariant. Next, $\markov$ is a subalgebra. So if $A \in \markov$, then $[A, \projm dH(A)]\in \markov$ and $\markov$ is also invariant. 

\end{proof}

\begin{remark}
\label{rem:poisson}
The following subtlety should be noted in (b).  $\dd^\perp$ and $\markov^\perp$ are Poisson subspaces, but $\dd$ and $\markov$ are not. A subspace $V \subset \gln$ is a Poisson subspace if and only if the restriction of every Hamiltonian vector field on $\gln$ to $V$ is tangent to $V$~\cite[Prop. 3.33]{AdvM-book}.  However, it is only vector fields corresponding to the quadratic Hamiltonians \qref{eq:quad_hamilton} that vanish on $\dd$ and $\markov$.
\end{remark}
\subsection{Convexity of $f$ and spectral negativity}
We defined the Markov group in order to give the Lax equations \qref{eq:discrete} the Hamiltonian structure stated in Theorem~\ref{thm:hamilton}.  We now return to the probabilistically interesting case.  We consider the initial value problem for \qref{eq:discrete} with initial data that are generators of spectrally negative Markov processes with discrete states $-\infty < u_1 < \ldots < u_N < \infty$. In this discrete setting, spectral negativity simply means that the generator is a lower-triangular matrix. For brevity, let $\lotri$ denote the algebra of lower-triangular matrices. We then have
\begin{theorem}
\label{thm:convex}
Assume $f$ is a $C^1$ convex flux and $F$ is defined by \qref{eq:Markov2}. Then
\begin{enumerate}
\item[(a)] The vector field \qref{eq:discrete} leaves $\lotri$ invariant.	
	
\item[(b)] The vector field \qref{eq:discrete} leaves $\cone \cap \lotri$ positively invariant.
\item[(c)] For every $A_0 \in \cone \cap \lotri$, there is unique, global, $C^\infty$ solution $A: [0,\infty)$ with initial condition $A(0)=A_0$.
\end{enumerate}
\end{theorem}
\begin{proof}
(a) is easy. If $A \in\lotri$, then $B \in \lotri$ and $[A,B] \in \lotri$. (This does not require convexity of $f$).

\noindent
(b)  Since we have already shown that the lower-triangular form is preserved, we only need to show that the flow preserves positivity. In order to establish positive invariance, we show that the vector field \qref{eq:discrete} `points into' $\cone \cap \lotri$ for a point on its boundary. A point on the boundary has $A \in \cone \cap \lotri$ and $A_{ij}=0$ for some indices $(i,j)$ in the lower-triangular region (i.e $i >j$). We claim that $\dot{A}_{ij} \geq 0$ so that positivity is preserved. Indeed, by \qref{eq:discrete}
\be
\label{eq:pos1}
\dot{A}_{ij} = \sum_{k \neq ij} \left(F_{kj} -F_{ik}\right) A_{ik}A_{kj} + \left(A_{ii}-A_{jj}\right)B_{ij} + \left(B_{jj}-B_{ii}\right) A_{ij}.
\ee
The second and third term vanish since $A_{ij}=0$ and $B_{ij}=F_{ij}A_{ij}=0$. We now consider the first term. Since $A$ is lower-triangular, the sum extends only over $j < k < i$. For such $k$, $u_j < u_k < u_i$, and the convexity of $f$ and \qref{eq:Markov2} imply
\be
\label{eq:pos2}
F_{kj} -F_{ik} = \frac{f(u_i)-f(u_k)}{u_i-u_k} - \frac{f(u_k)-f(u_j)}{u_k-u_j}  \geq 0.
\ee
Finally, since $A \in \cone \cap \lotri$ we also have $A_{ik}\geq 0$ and $A_{kj} \geq 0$. Thus, $\dot{A}_{ij} \geq 0$.

\noindent
(c) The vector field \qref{eq:discrete} vanishes on the diagonal, so that the diagonal terms are conserved. The diagonal terms control the off-diagonal terms since $A \in \cone$. Thus, we have the uniform bound

\be 
\label{eq:pos3}
0 \leq A_{ij}(t) \leq |A_{ii}(t)| = |A_{ii}(0)|, \quad i >j, t \geq 0.
\ee

\end{proof}

\begin{remark}
As in Remark~\ref{rem:poisson}, $\lotri$ constitutes an invariant subspace, but not a Poisson subspace. In fact, if we set $F_{ii}=0$ (which does not affect the flow), we see that the Hamiltonian $H= \Tr(A F\circ A)/2$ vanishes on the subalgebra of lower triangular Markov matrices. 
\end{remark}

\subsection{A principle of least action}
\label{sec:geodesic}
In this section, we assume that $F_{ij}$ does not vanish. Equation~\qref{eq:Markov3} then defines an endomorphism of $\markov$, and we have the following principle of least action. Let $\group: [0,1] \to \Markov$ denote a $C^1$ path with given endpoints.  Assume $F_{ij} \neq 0$, define the endomorphism $I: \markov \to \markov$, $B \mapsto A$ given by \qref{eq:Markov3}, and define the action 
\be
\label{eq:geod6}
\action[\group] = \frac{1}{2}\int_0^1 \Tr\left(I\left(\group^{-1}\dot{\group}\right) \group^{-1}\dot{\group}\right) \, dt.
\ee
Then the Euler-Lagrange equations for minimizing this action are 
\be
\label{eq:geod5}
\dot{A}=[A,B], \quad \dot{\group} = \group B.
\ee
The second equation is a linear non-autonomous equation and can be integrated once we have solved the first, which is of course, \qref{eq:discrete}.

It is easy to check this assertion, but the main point to observe is that 
when $F_{ij}>0$, a similar endomorphism of $\son$ is used to define geodesic flow with respect to a left-invariant metric on $SO(N)$. Let us first recall these ideas. This will immediately explain the origin of \qref{eq:geod6}. 

We follow the notation of~\cite[Ch.8, p.265]{AdvM-book} (see also~\cite[Appendix 2]{Arnold}).  Assume $F$ is a symmetric matrix with strictly positive entries and let $I: \son \to \son$ denote the endomorphism  $\omega \mapsto I(\omega)$ with $I(\omega)_{ij} =\omega_{ij}/F_{ij}$. In physical terms, $\omega$ is the angular momentum in the body frame and $X= I(\omega)$ is the angular momentum in the body frame. The letter $I$ stands for the inertia tensor. Since $F_{ij}>0$ and $\son$ is semi-simple, the quadratic form
\be
\label{eq:geod1}
\left( \omega_1, \omega_2\right)_F:= \Tr\left( I(\omega_1) \omega_2 \right) , \quad \omega_1,\omega_2 \in \son
\ee
is an inner-product on $\son$. The inner-product $(\cdot, \cdot)_F$ then defines a left-invariant metric on $SO(N)$ by left-translation. The length of a $C^1$ path $\group: [0,1] \to SO(N)$ with respect to this metric is given by 
%The subscript $f$ in \qref{eq:geod1} reminds us that $F$, hence $I$, are obtained from a flux $f$; thus it is $f$ that defines the metric on $SO(N)$. 
\be
\label{eq:geod2}
\length[\group] = \int_0^1 \sqrt{\left(\group^{-1} \dot{\group},\,\group^{-1}\dot{\group}\right)_F } \, dt. 
\ee
The problem of minimizing the length is the same as that of minimizing the action
\be
\label{eq:geod3}
\action[\group] = \frac{1}{2}\int_0^1\left(\group^{-1}\dot{\group},\, \group^{-1}\dot{\group}\right)_F \, dt.  
\ee
The Euler-Lagrange equations for this variational principle are precisely the Euler equations:
\be
\label{eq:geod4}
\dot{X} =[X, \omega], \quad \dot{\group} = \group \omega.
\ee
If we replace the group $SO(N)$ with the Markov group $\Markov$, the angular momentum $X$ with $A$, and the angular velocity $\omega$ with $B$, then we have a flow on $\Markov$ given by precisely \qref{eq:geod5}.
The analogy with geodesic flow on $\son$ is now clear.

In the calculations above, we did not need to assume that $F_{ij}$ is of the form \qref{eq:Markov2}. Indeed, every symmetric matrix $F$ with positive entries defines a diagonal left-invariant metric on $SO(N)$. Geodesic flow with this metric is Hamiltonian, but not necessarily completely integrable for $N \geq 4$. However, Manakov discovered that when
\be
\label{eq:geod8}
F_{ij} = \frac{f_i -f_j}{u_i -u_j}
\ee
for two vectors $(f_1,\ldots, f_N)$, and $(u_1, \ldots, u_N)$, the geodesic flow is integrable.  It is quite remarkable that we find ourselves in exactly this situation with the flux function $f$ in \qref{eq:sc1} defining $F$ as in \qref{eq:Markov2}. 

The analogy with geodesic flow is incomplete in the following respect. First, we do not assume that $F_{ij}>0$. Moreover, even when $F_{ij}>0$, the metric on $\markov$ is degenerate since $\markov$ is not semi-simple. (The 
Killing form of $A,B \in \gln$ is $ 2N\Tr(AB)-2 \Tr(A) \Tr(B)$. This vanishes on the identity in $\gln$, and on $\markov$ if  $A= \sum_{j=1}^N c_j E_j$ with $\sum_{j=1}^N  c_j=0$.)

\section{Complete integrability}
\label{sec:aks}
Complete integrability of all the systems alluded to in  Section~\ref{subsec:turb} can be established in the unified framework of~\cite{AdvM1,AdvM2,RSTS}. There are two distinct aspects to these studies: the first is to establish complete integrability via a suitable loop algebra splitting.  The second is to explicitly linearize the flow on a Jacobi variety. Here we only consider the first aspect of the problem. We show that the Lax equation \qref{eq:discrete} defines a completely integrable Hamiltonian system. The proof is almost a textbook application of 
the Adler-Kostant-Symes (AKS) theorem  and we follow the treatment in~\cite[\S 4.4]{AdvM-book}. Construction of the linearizing transformation is more difficult and will be considered in a separate article.

\subsection{Integrability via the AKS theorem}
We introduce the loop algebra of formal finite Laurent expansions valued in $\gln$
\be
\label{eq:loop1}
\loopalg = \{ X(z) = \sum_{m}^n A_k z^k , m,n \in \Z\;\; A_k \in \gln\}.
\ee
The natural Lie bracket on $\loopalg$ is given by
\be
\label{eq:loop_bracket}
\left[ \sum_{i \leq n} A_i z^i, \sum_{j \leq m} B_j z^j\right] = \sum_{k \leq m+n} z^k \left( \sum_{i+j=k} [A_i,B_j]\right).
\ee
The sum includes only a finite number of terms by \qref{eq:loop1}. We pair $\gln$ with $\gln^*$ via the non-degenerate, $\Ad$-invariant pairing \qref{eq:ad-pair}. There are then various $\Ad$-invariant pairings that one may introduce on $\loopalg$. We use the pairing
\be
\label{eq:loop_pairing}
\left< X | Y \right> = \sum_{i +j =0} (X_i, Y_j) = \frac{1}{2\pi i}\oint_{|z|=1} \Tr( X(z) Y(z)) \frac{dz}{z}.
\ee
It is easily checked that this pairing is non-degenerate and $\Ad$-invariant.

The direct sum decomposition \qref{eq:decomp1} also induces a decomposition of $\loopalg$. We define the subalgebras
\ba
\label{eq:loop2}
\loopp & =& \{ X(z) = \sum_{k \geq 0} A_k z^k , \quad A_0 \in \markov, \,\, A_k \in \gln, k \geq 1 \},\\
\label{eq:loop3}
\loopm & =& \{ X(z) = \sum_{k \leq 0} A_k z^k , \quad A_0 \in \dd, \,\, A_k \in \gln, k \leq -1 \}.
\ea
It is immediate from the calculations of Section~\ref{sec:algebra} that
\be
\label{eq:loop4}
\loopalg = \loopp \oplus \loopm.
\ee
The respective projections are given by
\be
\label{eq:loop5}
X(z)_+ = \projm A_0+ \sum_{k \geq 1} A_k z^k, \quad X(z)_-= \projd A_0 + \sum_{k \leq -1} A_k z^k.
\ee
The orthogonal complements with respect to the pairing \qref{eq:loop_pairing} are given by 
\ba
\label{eq:loop6}
\loopp^\perp & =& \{ Y(z) = \sum_{k \geq 0} Y_k z^k , \quad Y_0 \in \markov^\perp, \, Y_k \in \gln, k \geq 1\},\\
\label{eq:loop7}
\loopm^\perp & =& \{ Y(z) = \sum_{k \leq 0} Y_k z^k , \quad Y_0 \in \dd^\perp, \, Y_k \in \gln, k \leq -1 \}.
\ea
The gradient of a function $H: \loopalg \to \C$ is defined through the pairing \qref{eq:loop_pairing}. For $X,Y \in \loopalg$
\be
\label{eq:loop_lax}
\left< \nabla H(X) | Y \right> = \frac{d}{d\tau}H(X+\tau Y) |_{\tau=0}. 
\ee
Hamiltonian flows on $\loopalg$ correspond to the Lax equation 
\be
\dot{X}=[\nabla H(X), X]
\ee
If $H$ is $\Ad$-invariant then the vector-field \qref{eq:loop_lax} vanishes. On the other hand, $\Ad$-invariant Hamiltonians define non-trivial vector fields through the $r$-matrix induced by the splitting~\qref{eq:loop4}.  By the Adler-Kostant-Symes theorem, these vector fields correspond to the Lax equation
\be
\label{eq:loop_lax_R}
\dot{X}= \pm[X, \nabla H(X)_\mp].
\ee

We now show that \qref{eq:discrete} is of the form \qref{eq:loop_lax_R} for a suitable $\Ad$-invariant Hamiltonian on $\loopalg$. Let $f$ denote the flux in the scalar conservation law \qref{eq:sc1}, and consider its antiderivative
\be
\label{eq:def_ham1}
F(s) =  \int_0^s f(r) dr.
\ee
Define the Hamiltonian $H_F: \loopalg \to \C$
\be
\label{eq:def_ham2}
H_{F}(X(z)) = \left< F(X(z) z^{-1})| z^2 \right> =  \frac{1}{2\pi i}\oint_{|z|=1} F( X(z)z^{-1}) z \, dz.  
\ee
If $f$ is a polynomial, the second equality follows from Cauchy's theorem. The general case follows by approximation. The Hamiltonian $H_F$ is {\em distinct\/} from the Hamiltonian of Theorem~\ref{thm:hamilton}, and in some sense is more natural.
A few calculations (see ~\cite[p.94]{AdvM-book}) then yield that $H_F$ is $\Ad$-invariant and
\be
\label{eq:ham_loop1}
\nabla H_F(X) = f(Xz^{-1}) z.
\ee
Now consider the diagonal matrices $\My$ and $\Mf$ as in \qref{eq:diagonals}, and consider the finite-dimensional subspace $V$ of $\loopalg_+$ consisting of linear polynomials of the form
\be
\label{eq:ham_loop2}
V=\{ X \in \loopalg_+| X(z) = z\My - A\}.
\ee
A direct computation based on the definition of $H_F$ then yields
\be
\label{eq:ham_loop3}
\left(\nabla H_F(X)\right)_+ = z \Mf + B
\ee
where $B$ is as in \qref{eq:Markov3}. Thus, the Hamiltonian flow defined by $H_F$ on $V$ is exactly \qref{eq:manakov2}, which is of course, identical to \qref{eq:discrete}.

We now see that every $C^1$ flux $f$ gives rise to a Hamiltonian flow as in \qref{eq:manakov2}. Since each of these Hamiltonians is $\Ad$-invariant, they are all in involution. One may now count the number of integrals and invoke the Liouville theorem. Here we linearize the flow via the AKS theorem. 
\begin{theorem}
\label{thm:ci}
Let $A_0 \in \markov$ and $X_0(z)= z\My -A_0$. Let $g_\pm(t)$ denote the smooth curves in $\Loopgroup_\pm$ which solve the factorization problem
\be
\label{eq:factorization}
\exp(-t \nabla H_F(X_0)) = g_+(t)^{-1} g_-(t), \quad g_\pm(0)=\Id,
\ee
for $t$ in a maximal open interval $I$ containing $0$. Then the solution $X_t$, $t \in I$ to \qref{eq:manakov2} with initial condition $X_0$ is given by 
\be
\label{eq:factor_sol}
X_t = \Ad_{g_+(t)} X_0 = \Ad_{g_-(t)} X_0.
\ee
\end{theorem}
Note that a solution to the factorization problem always exists for $|t|$ small, thus  a maximal interval of existence for a smooth solution to the factorization problem is well-defined. It is also well-known that this factorization problem is equivalent to a Riemann-Hilbert problem (see~\cite[Ch. 3.5]{Babelon}). 

\subsection{An operator factorization problem}
While Theorem~\ref{thm:ci} applies only to the discrete Lax equations \qref{eq:discrete}, one may easily guess the associated factorization problem for \qref{eq:kinetic}.  Recall that the generators $\gena$ and $\genb$ are integro-differential operators defined by \qref{eq:gena1} and \qref{eq:genb}. As $N \to \infty$, the $N \times N$ diagonal matrices $\My$ and $\Mf$ of Theorem~\ref{thm:ci} are replaced by multiplication operators that act on test functions via
\be
\label{eq:op1}
\My \varphi (u) = u \varphi(u), \quad \Mf \varphi(u) = f(u)\varphi(u).
\ee
The crucial algebraic relation \qref{eq:manakov} continues to hold.
\be
\label{eq:op2}
[\gena, \Mf]-[\My, \genb]=0.
\ee
Formally, this is all that is required to embed \qref{eq:kinetic} in a loop-group and we now find the factorization problem
\be
\label{eq:factorization2}
\exp(-t \nabla H_F(X_0)) = g_+(t)^{-1} g_-(t), \quad g_\pm(0)=\Id,
\ee
with $X_0 = z\My - \gena$. Rather than develop these ideas in formal generality, let us mention one interesting example. Assume we consider Burgers equation with Brownian motion initial data. Then the Hamiltonian is $H_F(s) = s^3/6$ and the generator of initial data is  $\gena_0 \varphi(u) = - \varphi''(u)/2$. Then we find
\be
\label{eq:factorization3}
\nabla H_F(X_0) = \frac{1}{2}\left(z \My  - \frac{1}{2}\frac{d^2}{du^2}\right)^2.
\ee
For $z \in \R$, this is the square of the Airy operator. To the best of our knowledge, the factorization suggested by \qref{eq:factorization2} is new.

\section{Scattering and inverse scattering theory}
\label{sec:scatter}
In this section we develop a scattering and inverse scattering theory for the discrete zero-curvature equations~\qref{eq:discrete-zc}. 
\be
\nonumber
\partial_t{A} -\partial_x B =[A,B].
\ee
The linear problem that underlies the scattering theory of \qref{eq:discrete-zc} is as follows. Assume that $\My$ is a fixed diagonal matrix as in \qref{eq:diagonals} and 
$A \in L^1 (\R, \markov)$. We consider a fundamental matrix for the linear equation
\be
\label{eq:scatter1}
\psi_x = \psi \left(z\My + A \right), \quad x \in \R,
\ee
such that $\psi(x,z) \sim e^{zx\My}$ as $x \to -\infty$. Such solutions are called {\em wavefunctions\/} (we use the terminology of~\cite{TU1}). The scattering theory for this equation when $A$ is a matrix that vanishes on the diagonal was considered by Zakharov~{\em et al}~\cite{Zakharov} and by Beals and Coifman~\cite{BC1,BC2}.  In our work, $A \in \markov$. As a consequence, even though $A$ is completely determined by its off-diagonal elements, its diagonal entries do not vanish. Thus, the scattering theory of~\cite{BC1,BC2} does not immediately apply to our model and we have to rederive some results. For the most part, this is straightforward. To prevent too much repetition, we state the results we need, and
present the main calculations that explain how the results of~\cite{BC1,BC2} are to be modified in Section~\ref{subsec:bc1} and Section~\ref{subsec:bc2}.

A short outline of the results of this section is as follows. The scattering theory is addressed in Section~\ref{subsec:scatter}. Theorem~\ref{thm:bc1} and Theorem~\ref{thm:bc2} associate spectral data to $A \in L^1(\R, \markov)$.
We consider the  inverse scattering theory in Section~\ref{subsec:inv} and 
 Theorem~\ref{thm:inv-spec}. 
The time evolution of spectral data and the Cauchy problem is considered in Section~\ref{subsec:cauchy}.
The linear evolution of spectral data is stated in~\qref{eq:spec-evolve}. The combination of  inverse scattering theory and evolution yields several well-posedness theorems for \qref{eq:discrete-zc}. 
Finally, we construct a hierarchy of integrable models as in the ZS-AKNS hierarchy in Section~\ref{subsec:ZS-AKNS}. It is not apparent to us that these have intrinsic probabilistic significance, but it is interesting to note that one may construct other integrable flows on $\Markov$ with little effort.  

All results are rigorous for $N \times N$ matrices and have a natural, but formal, extension to the integro-differential operators $\gena$ and $\genb$.  The linear evolution of the spectral data for these operators remains \qref{eq:spec-evolve} with $\My$ and $\Mf$ replaced by the multiplication operators~\qref{eq:op1}. However, we are unaware of rigorous results on the inverse spectral problem for such operators, and a full well-posedness theorem for \qref{eq:zc} via inverse scattering requires further study.

\subsection{Scattering theory}
\label{subsec:scatter}
It is more convenient to work with the new variable
\be
\label{eq:scatter-2}
m(x,z) = e^{-z\My x} \psi(x,z).
\ee
$m$ satisfies the linear equation
\be
\label{eq:scatter3}
m_x = z[m,\My] + m A
\ee
Solutions to \qref{eq:scatter3} such that $\|m(\cdot,z)\|_{L^\infty(\R)} < \infty$ and $m(x,z) \to \id$ as $x \to -\infty$ are called {\em global reduced wave functions\/}.  

Theorem~\ref{thm:bc1} and Theorem~\ref{thm:bc2} below closely follow Beals and Coifman~\cite{BC1}.
Let $\Sigma = i\R$ denote the imaginary axis in the complex $z$-plane, 
and let $\pot$ denote the set of maps $A \in L^1(\R,\markov)$. We call these maps  {\em potentials\/}  and the subclass $\pot_0$ below {\em generic potentials\/}. 

\begin{theorem}
\cite[Thm. A]{BC1}
\label{thm:bc1}(a) Suppose $A \in \pot$. There is a bounded discrete set $Z \subset \C\backslash\Sigma$ such that $m(\cdot,z)$ is a unique global reduced wave function for every $z \in \C\backslash(\Sigma \cup Z)$. Moreover, $m(x, \cdot)$ is meromorphic in $\C \backslash \Sigma$ with poles precisely at the points of $Z$ and $\lim_{z \to \infty} m(x,z)=\id$.

(b) There is a dense open set $\pot_0 \subset \pot$ such that if $A \in \pot_0$ then
\begin{enumerate}
\item $Z$ is finite.
\item The poles of $m(x, \cdot)$ are simple.
\item Distinct columns of $m(x,\cdot)$ have distinct poles.
\item $m(x,\cdot)$ admits limits $m^\pm(x,)$ as $z \to \Sigma \backslash Z$ from the left and right half-planes.
\end{enumerate}
\end{theorem}
\begin{theorem}
\cite[Thm. B]{BC1}
\label{thm:bc2}
(a) Suppose $A \in \pot_0$. For $z \in \Sigma$ there is a unique matrix $v(z)$ such that for every $x \in \R$
\be
\label{eq:bc-jump}
m^+(x,z) = e^{-xz\My} v(z) e^{xz\My}m^-(x,z).
\ee

(b) For each pole $z_j \in Z$, there is a matrix $v(z_j)$ such that the residue of $m$ satisfies
\be
\label{eq:bc-residue}
\Res(m(x, \cdot);z_j) =\lim_{z \to z_j} e^{-xz\My} v(z) e^{xz\My}m(x,z).
\ee

(c) The generic potential $A$ is uniquely determined by the jump matrix $v(z)$, $z \in \Sigma$ and the residues $\Res(m(x, \cdot);z_j)$, $z_j \in Z$.
\end{theorem}
\subsection{Inverse scattering theory}
\label{subsec:inv}
The jump matrix $v(z)$, the poles $Z=\{z_1, \ldots, z_M\}$ and the residues $v(z_j)$ constitute the {\em scattering data\/}.  The reconstruction of $A$ from the scattering data is the {\em inverse spectral problem}. Though Theorem~\ref{thm:bc2} guarantees that the scattering data associated to a generic potential is unique, this assertion is proved via an application of Liouville's theorem and is not constructive. What is required is a constructive procedure to obtain $A$ given the scattering data. 

In order to state the inverse spectral theorems, we work with potentials that lie in the Schwartz class $\schwartz(\R, \markov)$. This assumption is not necessary, but it simplifies the exposition. Analogous finite regularity results can also be obtained as in~\cite{BC1}. 
\begin{theorem}
\label{thm:inv-spec}
(a) Suppose $A \in \schwartz(\R,\markov)$. Then there is $R>0$ and $C^\infty$ functions $m\ksup: \R \to \gln$, $k=0,1, \ldots$ such that  
\be
\label{eq:bc-asympt}
m(x,z) = \sum_{k=0}^\infty z^{-k} m\ksup(x), \quad x \in \R, |z| >R, 
\ee
and the series converges uniformly in $x$ and $z$.

(b) The coefficients $m\ksup$ may be determined recursively. In particular, $m^{(0)}$ is a diagonal matrix with entries
\be
\label{eq:scatter-28}
m^{(0)}_{ii}(x) = \exp \left( \int_{-\infty}^x A_{ii}(s) \, ds\right), \quad i=1, \ldots, N,
\ee
and the off-diagonal entries of $m^{(1)}$ are given by
\be
\label{eq:scatter-29}
m^{(1)}_{ij} = \frac{m^{(0)}_{ii}}{u_j-u_i} A_{ij}(x),\quad i \neq j.
\ee
(c) The asymptotic expansion \qref{eq:bc-asympt} may also be written 
\be
\label{eq:bc-asympt2}
m(x,z) = m^{(0)}(x) h(x,z), \quad h(x,z) = \sum_{k=0}^\infty z^{-k} h\ksup(x), \quad x \in \R, |z| >R, 
\ee
where $h^{(0)}(x) \equiv \id$, and $h\ksup$, $k \geq 1$, are in the Schwartz class $\schwartz(\R, \gln)$.
\end{theorem}
Recall that here $u_1 < u_2 < \ldots < u_N$ are the diagonal entries of $\My$. Part (b) of the theorem  allows us to uniquely reconstruct the potential. Assume given a global reduced wave function $m$ with the asymptotic expansion \qref{eq:bc-asympt}. Then the off-diagonal terms of $A$ are given by \qref{eq:scatter-29}, and the diagonal terms are given by the relation $A_{jj}=-\sum_{k \neq j} A_{jk}$. It is necessary to assume that $m$ is a reduced wave function: an arbitrary set of functions $m\ksup(x)$ is not admissible. Indeed, the constraint $A_{jj}=-\sum_{k \neq j}A_{jk}$ implies many relations between the coefficients $m\ksup$. For example, we have 
\be
\label{eq:scatter-29b}
m^{(0)}_{ii}(x) = \sum_{j \neq i} \int_{-\infty}^x m^{(1)}_{ij}(s)\, ds. 
\ee

The full inverse scattering problem relates the scattering data to the potential $A$. 
In light of Theorem~\ref{thm:inv-spec}, it suffices to reconstruct $m$ from the scattering data. Since $m$ is holomorphic in $\C \backslash(\Sigma \cup Z)$ it is expressed in terms of the scattering data by Cauchy integrals. The associated integral equations are independent of our assumption that $A \in \markov$, and the results of~\cite{BC1} relating $m$ and the scattering data apply directly. The subtlety is that the scattering data satisfies both algebraic, analytic and topological constraints. For example, these may be constraints involving the zeros and winding numbers of principal minors of $v$~\cite[Thm D]{BC1}. For generic potentials that satisfy these constraints, $m$ and the scattering data are related by Cauchy integrals that preserve the Schwartz class. For such scattering data, the inverse scattering problem is solved by mapping the scattering data to $m$ via Cauchy integrals, and then $m$ to $A$ via Theorem~\ref{thm:inv-spec}. In the simplest situation, $Z$ is empty, and the wave function is reconstructed from the jump on $\Sigma$ alone.

\subsection{Evolution of scattering data and the Cauchy problem}
\label{subsec:cauchy}

We now combine the $x$ and $t$ dependence, and consider the Cauchy problem for the discrete zero-curvature equations~\qref{eq:discrete-zc} with initial data $A(x,0)=A_0(x) \in \markov$. Let $v_0(z)$, $z \in \Sigma$ and $v_0(z_j)$ denote the scattering data of $A_0$.  The scattering data evolve by the simple linear equations
\ba
\label{eq:spec-evolve}
\frac{\partial v(z)}{\partial t} &=& [z\Mf, v(z)], \quad z \in \Sigma \\
\frac{\partial v(z_j)}{\partial t} &=& [z\Mf, v(z_j)], \quad z_j \in Z.
\ea
with the unique solution
\be
\label{eq:spec-evolve2}
v(z,t) = e^{t z \Mf}v_0(z) e^{-tz\Mf}, \quad v(z_j,t) = e^{t z_j \Mf}v_0(z_j) e^{-tz_j\Mf}.
\ee
The evolution is {\em formally stable\/} in the terminology of~\cite[\S 3.12]{BC2}). We may now combine Theorem A and Theorem C of~\cite{BC2} to obtain the following basic well-posedness theorem for \qref{eq:discrete-zc}.
\begin{theorem}
\label{thm:well-posed1}
Assume $A_0$ is a generic potential in $ \schwartz(\R,\markov)$ with associated scattering data $v_0(z)$, $z \in \Sigma$ and $v_0(z_j)$, $z \in Z$. Then there is $T>0$ and a unique smooth map $[0,T) \to \schwartz(\R,\markov)$, $t \mapsto A(\cdot,t)$, such that the scattering data of $A(\cdot,t)$ is given by \qref{eq:spec-evolve2} and $A(x,t)$ solves the Cauchy problem for the discrete zero-curvature equations \qref{eq:discrete-zc} with initial data $A_0$.
\end{theorem}
In addition, the following dichotomy holds~\cite[Thm. B]{BC2}.
\begin{theorem}
\label{thm:well-posed2}
Under the hypotheses of Theorem~\ref{thm:well-posed1}, suppose $T \in (0,\infty]$ is maximal. Then either $T=\infty$ or $\lim_{t \to T} \| A(\cdot,t)\|_{L^2(\R)}=\infty$.
\end{theorem}
In general, the maximal time interval is finite. However, global existence is guaranteed if $A_0$ is triangular. If $A$ is triangular, so are $B$ and $[A,B]$, and  $[A,B]$ vanishes on the diagonal. Thus
\be
\label{eq:bc-trace}
\Tr(A^T [A,B]) =0.
\ee
In addition, for $A \in \schwartz(\R, \markov)$
\be
\int_\R \Tr \left(A^T(x) \partial_x B\right) \, dx= \sum_{i,j} F_{ij}\int_\R A_{ij} \partial_x A_{ij}\, dx =0. 
\ee
Thus, the evolution of \qref{eq:discrete-zc} is {\em dissipative\/} (see~\cite[\S 1.11]{BC2}) and we have global existence.
\begin{corollary}
\label{cor:global}
Assume the hypotheses of Theorem~\ref{thm:well-posed1} and assume in addition that $A_0(x)$ is triangular for every $x \in \R$. Then $T=\infty$.
\end{corollary}
It is surprising that we do not need to assume that $A$ is everywhere lower triangular or everywhere upper triangular. The assumption that $A$ is triangular pointwise is enough to ensure \qref{eq:bc-trace} for every $x \in \R$, which in turn implies dissipativity. 

Finally, let us connect these results with the probabilistic context that motivated us. In order to ensure that $A$ is truly a generator, we must ensure that the off-diagonal terms are positive. Since smooth solutions exist, this is preserved at least for a short time. However, it is more subtle to ensure global existence. Here the convexity of $f$ plays an important role. 
\begin{theorem}
\label{thm:positivity}
Assume the hypotheses of Theorem~\ref{thm:well-posed1}. In addition, assume that $f$ is convex, and that $A_0(x)$ is the generator of a spectrally negative Markov process. Then $T=\infty$ and $A(x,t)$ remains the generator of a spectrally negative Markov process for every $t >0$.
\end{theorem}
\begin{proof}
Equation \qref{eq:discrete-zc} may also be solved by the method of characteristics. Indeed, $B_{ij}=F_{ij}A_{ij}$, thus each entry $A_{ij}$ evolves on a characteristic with speed $-F_{ij}$. The characteristic speed is simply the Rankine-Hugoniot condition associated to the shock connecting states $u_i$ and $u_j$. We integrate \qref{eq:discrete-zc} on characteristics to 
find
\be
\label{eq:chars}
A_{ij}(x,t) = (A_0)_{ij}(x+F_{ij}t) + \int_0^t [A,B]_{ij}(x+F_{ij}s,s)\, ds.
\ee
Since $u_1 < \ldots < u_M$ and $f$ is convex, we now find exactly as in the proof of Theorem~\ref{thm:convex} that $[A,B]_{ij}  \geq 0$, $i >j$. The diagonal terms are conserved on characteristics since $[A,B]_{ii}=0$ since $B= F \circ A$. Similarly, the upper-triangular part of $[A,B]$ vanishes if $A$ and $B$ are lower-triangular. A  simple maximum principle argument shows that $A$ remains lower-triangular.
\end{proof}
The integral equations \qref{eq:chars} can also be used to give a direct proof of global existence of solutions without the assumption that $A_0$ is generic. Since $A_{ij}\geq 0$ on the off-diagonal, and $A_{ii}$ is conserved, we have the bound $\|A_{ij}(\cdot, t)\|_{L^\infty(\R)} \leq \|(A_0)_{ii}\|_{L^\infty(\R)}$, which ensures global existence.

\subsection{The ZS-AKNS hierarchy}
\label{subsec:ZS-AKNS}
The discrete zero-curvature equations \qref{eq:discrete-zc} are part of a hierarchy of commuting Hamiltonian flows. The existence of such hierarchies was established in the pioneering work of Zakharov and Shabat~\cite{ZS1} and Ablowitz, Kaup, Newell and Segur~\cite{AKNS}. We now derive the associated hierarchy for \qref{eq:discrete-zc}: as expected this is a modification of the hierarchy for the $N$-wave model. Our calculations and notation follow~\cite{TU1}.

We fix a diagonal matrix $\Mf$ and consider the asymptotic behavior of $Q(x,z)= m^{-1}\Mf m$ as $z \to \infty$. By Theorem~\ref{thm:inv-spec}(c), we have $m^{-1}\Mf m = h^{-1}\Mf h$, and the expansion \qref{eq:bc-asympt2} yields
\be
\label{eq:hier1}
Q(x,z)= h^{-1}\Mf h \sim \sum_{k=0}^\infty Q\ksup z^{-k}, \quad z \to \infty.
\ee
We call $Q\ksup$ the $k$-th flux. It admits an expansion
\be
\label{eq:hier4}
Q\ksup= z^k \Mf + z^{k-1} B_1  + z^{k-2}B_2 +\ldots + B_k, \quad B_k \in \markov.
\ee
\begin{defn}
The $k$-th flow in the hierarchy is given by the equation
\be
\label{eq:hier-k}
\partial_t A - \partial_x Q\ksup = [A, Q\ksup], \quad k \geq 0.
\ee
\end{defn}
The zero-curvature equation \qref{eq:hier-k} may also be written in the Lax form
\be
\label{eq:hier-k-lax}
\left[ \partial_x + z\My + A, \partial_t + Q\ksup\right] =0.
\ee
The $k$-th flux is obtained as follows. We show below that $Q$ satisfies the linear equation
\be
\label{eq:hier3}
Q_x = [Q, z\My + A].
\ee
The asymptotic expansion \qref{eq:hier1} now yields the hierarchy of linear equations
\be
\label{eq:scatter-42o}
0 = [Q^{(0)}, \My],
\ee
and 
\be
\label{eq:scatter-42k}
Q\ksup_x - [Q\ksup,A] = [Q^{(k+1)}, \My], \quad k \geq 0.
\ee
Since $h(x,z) \sim \id + z^{-1}h^{(1)}$ as $z \to \infty$, we have
\be
\label{eq:hier5}
Q(x,z) \sim \Mf + \frac{[\Mf,h^{(1)}]}{z} + \ldots, \quad z \to \infty.
\ee
Thus, $Q^{(0)}=\Mf$ is the solution to \qref{eq:scatter-42o}. For $k\geq 1$ the ansatz \qref{eq:hier4} yields $k+2$ linear equations for $B_j$. We use \qref{eq:hier4} and \qref{eq:hier-k-lax} to obtain the equations
\ba
\label{eq:q1}
O(z^{k+1}): && [\My, \Mf]=0, \\
O(z^{k}):   && \partial_x \Mf + [A,\Mf] + [\My,B_1] =0, \\
O(z^j), k-1 \geq j \geq 1: && \partial_x B_{k-j} + [A, B_{k-j}] +[\My,B_{k+1-j}]=0, \\
O(1): && \partial_x B_k - \partial_t A + [A,B_k] =0.
\ea
The $O(z^{k+1})$ equation is trivially satisfied since $\Mf$ is diagonal. Since $\Mf$ is independent of $x$, the $O(z^k)$ equation generalizes equation  \qref{eq:manakov}. 
When $k=1$ this yields the off-diagonal terms of $B_1=B$ in accordance with \qref{eq:Markov3}. For $k >1$ we recursively solve \qref{eq:q1} until we obtain $B_k$. This is very similar to the recursion for the classical $N$-wave model with one important difference. When solving \qref{eq:q1} recursively, we realize that the $O(z)$ term yields only the off-diagonal terms of $B_k$. In the earlier work of Zakharov and Manakov, the diagonal terms of $B_k$ vanished because of the assumption that $B_k \in \unn$. Here these terms suffice to determine $B_k$ since $B_k \in \markov$.

The $k$-th flow may be solved by the inverse scattering method. The scattering data evolves by 
\be
\label{eq:spec-evolve3}
v(z,t) = e^{t z^k \Mf}v_0(z) e^{-tz^k\Mf}, \quad v(z_j,t) = e^{t z^k_j \Mf}v_0(z_j) e^{-tz^k_j\Mf}.
\ee
Well-posedness of the $k$-th flow requires that the evolution is formally stable~\cite[3.12]{BC2}. In our case, this requires $\Re(z^k\Mf_{jj})=0$ for $z \in \Sigma$ and each diagonal entry of $\Mf$. Since $\Sigma$ is the imaginary axis, this is satisfied for all odd $k$ when $\Mf$ is real (in particular, for $\Mf$ given by $\diag(f(u_1), \ldots, f(u_M))$. For even $k$, we need to choose $\Mf$ purely imaginary, and we find that $B_k$ is purely imaginary if $A$ is real. This is incompatible with $A \in \markov$. For odd $k$, Theorem~\ref{thm:well-posed1} and Theorem~\ref{thm:well-posed2} hold with \qref{eq:spec-evolve2} replaced by \qref{eq:spec-evolve3} and \qref{eq:discrete-zc} replaced by \qref{eq:hier-k}. It is not clear if these equations have a true probabilistic interpretation.

\subsection{Proofs of Theorem~\ref{thm:bc1} and Theorem~\ref{thm:bc2}}
\label{subsec:bc1}
We now present the calculations and matrix factorization theorem that underlie Theorem~\ref{thm:bc1} and Theorem~\ref{thm:bc2}. To this end, it is enough to assume that $A$ is $C^\infty$ with compact support in $x$. A density argument as in~\cite{BC1} yields the conclusions for $A\in L^1(\R, \markov)$. To construct a globally bounded solution $m(x,z)$ we assume that $A$ has compact support and solve the initial value problem
\be
\label{eq:scatter4}
\tilde{m}_x = z[\tilde{m},\My] + \tilde{m}A, \quad \tilde{m}(x,z) = \id, x \ll 0.
\ee
It is clear that $\qref{eq:scatter4}$ has a unique solution that is holomorphic in  $z$. In addition, since $\tilde{m}_x= z [\tilde{m},\My]$, $x \gg 0$, there exists a holomorphic matrix $s(z)$ such that 
\be
\label{eq:scatter5}
\tilde{m}(x,z) = \left\{ \begin{array}{ll} \id, & x \ll 0, \\ e^{-zx\My} s(z)e^{zx\My}, & x \gg 0. \end{array} \right. 
\ee
Observe also that $\tilde{m}$ is always invertible because
\be
\label{eq:scatter7}
\det(\tilde{m})(x,z) = \exp\left( \int_{-\infty}^x \Tr(A(s))\, ds\right).
\ee

We seek a bounded solution to \qref{eq:scatter3} of the form $m(x,z)=\tilde{g}(x,z)\tilde{m}(x,z)$. Since $m$ solves \qref{eq:scatter3} and $\tilde{m}$ solves \qref{eq:scatter4} we find that $\tilde{g}(x,z)= e^{-xz\My} g(z) e^{xz\My}$ for some matrix $g(z)$. The asymptotic behavior of $\tilde{m}$ in \qref{eq:scatter5} then implies
\be
\label{eq:scatter12}
m(x,z) = \left\{ \begin{array}{ll} e^{-xz \My}g(z) e^{xz\My}, & x \ll 0, \\ e^{-zx\My} g(z) s(z)e^{zx\My}, & x \gg 0. \end{array} \right. 
\ee
It only remains to choose $g$ so that $m$ is globally bounded in $x$ for fixed $z$. 

First assume  $z$ is in the left-half plane. Recall that $\My= \diag(u_1, \ldots, u_m)$ with $u_1 < u_2 < \ldots < u_m$. Thus,  $m_{jk}(x,z) = g_{jk}(z) e^{xz(u_k-u_j)}$ for $x \ll 0$. Since $\Re(z) >0$, $m_{jk}$ is bounded only if $g_{jk}=0$ for $j < k$. Thus, $g$ is {\em lower-triangular\/}.  We next find that $m$ is bounded  as $x \to \infty$ only if $g(z) s(z)$ is {\em upper-triangular\/}. Finally, since $m \to \id$ as $x \to -\infty$, we see that the diagonal entries of $g$ are all $1$. In order to choose $g$ in accordance with these constraints, recall that Gaussian elimination may be written as the matrix factorization
\be
\label{eq:scatter16}
s(z) =L(z) D(z) U(z)
\ee
where $L$ and $U$ are lower and upper triangular matrices that are $1$ on the diagonal and $D(z)=\diag(\det(s_1(z), \ldots, s_m(z))$ where $s_k(z)$ denotes the $k \times k$ upper-block of $s(z)$~\cite[Thm 1.1]{Faddeev}. By construction, $L$, $D$ and $U$ are
unique except at the zeros of $\det(s_k(z))$, $k=1, \ldots, m$. Since $s$ is entire, this set is discrete. We then choose $g^-{-}(z) = L(z)^{-1}$, the superscript denoting the left-half plane. By construction, $g$ is meromorphic in the left-half plane.

A similar calculation in the right-half plane $\Re(z)>0$ reveals that $g(z)$ must be upper-triangular and $g(z)s(z)$ must be lower-triangular. In this case, we factorize
\be
\label{eq:scatter18}
s(z) = \tilde{L}(z) \tilde{D}(z) \tilde{U}(z)
\ee
where now $\tilde{D}= \diag(\det(\tilde{s}_1(z), \ldots, \tilde{s}_m(z))$ and $\tilde{s}_k(z)$ denotes the $k \times k$ lower-block of $s(z)$. We now find $g^+(z)=\tilde{U(z)}^{-1}$. 

Let $Z$ denotes the set of zeros of $\det(s_k(z))$ and $\det(\tilde{s_k}(z))$, $k=1, \ldots,m$. The factorizations $g^\pm(z)$ are continuous on $z \in \Sigma \backslash{Z}$. To summarize, we have 
\be
\label{eq:scatter16a}
m^\pm(x,z) = e^{-xz\My} g^\pm(z) e^{xz\My} \tilde{m}(x,z), \quad \in \R, z \in \C\backslash{Z}
\ee
where $\tilde{m}(x,z)$ is entire in $z$, $g^\pm$ are obtained by factorizing $s$ as in \qref{eq:scatter16} and \qref{eq:scatter18}. In order to obtain the scattering data, we isolate the jump in $m$ on $\Sigma$.  If we set $\varphi (x,z)= m^+ (m^-)^{-1}$, $x \in \R$, $z \in \Sigma$, we find $\varphi_x =z[\varphi, \My]$. Thus, there exists a matrix $v(z)$ such that $\varphi(x,z) = e^{-xz\My} v(z) e^{xz\My}$ and we have
\be
\label{eq:scatter-23}
m^+(x,z) =e^{-xz\My} v(z) e^{xz \My} m^-(x,z), \quad x \in \R, z \in \Sigma \backslash{Z}.
\ee
These are the main calculations needed to establish the existence of the jump measure, and it is clear that the assumption $A \in \markov$ (as opposed to $A_{jj}=0$) has played only a minor role (e.g. \qref{eq:scatter7} has replaced $\det(\tilde{m})\equiv 1$.). The arguments in ~\cite[pp. 48-49]{BC1} are similarly modified to yield Theorem~\ref{thm:bc1} and Theorem~\ref{thm:bc2}.

\subsection{Proof of Theorem~\ref{thm:inv-spec}}
\label{subsec:bc2}
Theorem~\ref{thm:inv-spec} is a modification of~\cite[Thm 6.1]{BC1}. The main differences are that $m_0$ is no longer the identity, and we have to solve separately for diagonal and off-diagonal terms. It is simplest to postulate an expansion of the form \qref{eq:bc-asympt} and solve for the terms $m_j$. One may then justify the expansion for $A$ that is suitably regular as in~\cite[\S 6]{BC1}. 

Assume \qref{eq:bc-asympt} holds and $m(x,z) \to \id$ as $x \to -\infty$. We substitute this ansatz in \qref{eq:scatter3} to find the hierarchy of equations
\be
\label{eq:scatter-26o}
0= [m^{(0)}(x), \My],
\ee
and 
\be
\label{eq:scatter-26k}
\frac{dm\ksup}{dx} - m_k A= [m^{(k+1)}(x), \My], \quad k=0,1,\ldots
\ee
Equation \qref{eq:scatter-26k} implies that on the diagonal
\be
\label{eq:scatter-26kd}
\frac{dm\ksup_{ii}}{dx} - \left(m\ksup A\right)_{ii}=0.
\ee
On the off-diagonal
\be
\label{eq:scatter-26ko}
\frac{dm\ksup_{ij}}{dx} - \left(m\ksup A\right)_{ij}= \left(u_j-u_i\right) m^{(k+1)}_{ij}, \quad i \neq j.
\ee
Equation \qref{eq:scatter-26o} implies that $m^{(0)}$ is a diagonal matrix. We then solve \qref{eq:scatter-26kd} with $k=0$ to obtain \qref{eq:scatter-28}. It is simplest to solve the rest of the hierarchy by making the ansatz
\be
\label{eq:h-ansatz}
m(x,z) = m^{(0)}\sum_{k=0}^\infty z^{-k} h\ksup(x,z) = m^{(0)}h(x,z).
\ee
Let $A_d$ and $A_o$ denote the diagonal and off-diagonal terms of $A$. We substitute \qref{eq:h-ansatz} in \qref{eq:scatter3} to find
\be
\label{eq:scatter-32}
h_x = [h, z\My + A_d] + hA_o. 
\ee
Then we have the hierarchy of linear equations
\be
\label{eq:scatter-33}
0 = [h^{(0)}, \My],
\ee
and 
\be
\label{eq:scatter-34}
h\ksup_x - [h\ksup,A_d] - h\ksup A_o = [h^{(k+1)}, \My], \quad k \geq 0.
\ee
Equation \qref{eq:scatter-33} implies that $h^{(0)}$ is diagonal. We then consider the diagonal terms of \qref{eq:scatter-34} with $k=0$ to  find $h^{(0)}_x=0$. Since $\lim_{x \to -\infty}h(x,z)=\id$, this implies $h^{(0)}\equiv \id$ as expected. The off-diagonal terms of \qref{eq:scatter-34} with $k=0$ may be solved algebraically and yield
\be
\label{eq:scatter-36}
h^{(1)}_{ij}= \frac{A_{ij}}{u_i-u_j}, \quad i \neq j.
\ee

This process can be continued indefinitely. At each step, we first solve a differential equation that yields the diagonal terms of $h\ksup$, and then an algebraic equation that yields the off-diagonal terms of $h^{(k+1)}$. For example, we find 
\be
\label{eq:scatter-37}
h^{(1)}_{ii}(x) = \sum_{j \neq i}\frac{1}{a_i-a_j} \int_{-\infty}^x A_{ij}(s)A_{ji}(s)\, ds,
\ee
and for the off-diagonal terms of $h^{(2)}$
\be
\nn
%\label{eq:scatter-38}
h^{(2)}_{ij}=\frac{1}{u_i-u_j}\left( \sum_{k \neq i,j} \frac{A_{ik}A_{kj}}{u_i-u_k} + h^{(1)}_{ii} A_{ij} + \frac{A_{ij}(A_{ii}-A_{jj})}{u_i-u_j} - \frac{\partial_x A_{ij}}{u_i-u_j}) \right).
\ee
This process becomes increasingly unwieldy, but at every step $h\ksup$ is expressed as a finite number of integro-differential terms of $A$. If $A$ is in the Schwartz class so is $h\ksup$ for $k \geq 1$. This is enough to establish the uniform convergence of Theorem~\ref{thm:inv-spec}.

\section{Acknowledgements}
I thank Mark Ablowitz, Mark Adler, Percy Deift, Luen-Chau Li, David Mumford, Bob Pego, and Fraydoun Rezakhanlou for stimulating discussions. This work is part of a general program developed with Ravi Srinivasan and I am particularly indebted to him.
Thanks also to Michele Benzi for informing me of~\cite{Johnson}. This work was supported by NSF grant DMS 07-48482.

\bibliographystyle{siam}
\bibliography{m1}

\end{document}